\documentclass[12pt]{amsart}
\usepackage{amsmath}
\usepackage{amssymb}
\usepackage[latin1]{inputenc}
\usepackage[T1]{fontenc}
\newtheorem{lemma}{Lemma}
\newtheorem{example}{Example}
\newtheorem{remark}{Remark}

\newtheorem{proposition}{Proposition}

\newcommand{\tr}{{\rm Tr }}

\newcommand{\vp}{\varphi}

\newcommand{\C}{\mathbb{C}}
\newcommand{\Z}{\mathbb{Z}}
\newcommand{\N}{\mathbb{N}}

\newcommand{\R}{\mathbb{R}}
\newcommand{\be}{\begin{equation}}
\newcommand{\eeq}{\end{equation}}
\newcommand{\bet}{\begin{equation*}}
\newcommand{\eeqt}{\end{equation*}}
\newcommand{\bea}{\begin{eqnarray}}
\newcommand{\eeqa}{\end{eqnarray}}
\newcommand{\beat}{\begin{eqnarray*}}
\newcommand{\eeqat}{\end{eqnarray*}}

\newcommand{\h}[1]{\mathcal{#1}}
\newcommand{\hil}{\mathcal{H}}

\newcommand{\hA}{\mathcal{A}}
\newcommand{\hB}{\mathcal{B}}

\newcommand{\br}{\mathcal{B}(\R)}
\newcommand{\lh}{L(\mathcal{H})}

\begin{document}
\title[Semispectral measures as convolutions]{Semispectral measures as convolutions and their moment operators}

\author{Jukka Kiukas}
\address{Jukka Kiukas,
Department of Physics and Astronomy, University of Turku, FI-20014 Turku,
Finland.  At present:
 Institut f\"ur Mathematische Physik, TU Braunschweig,
DE-38106 Braunschweig, Germany.} \email{jukka.kiukas@utu.fi}
\author{Pekka Lahti}
\address{Pekka Lahti,
Department of Physics and Astronomy, University of Turku, FI-20014 Turku,
Finland} \email{pekka.lahti@utu.fi}
\author{Kari Ylinen}
\address{Kari Ylinen, Department of Mathematics, University of Turku,
FI-20014 Turku, Finland} \email{kari.ylinen@utu.fi}

\begin{abstract}
The moment operators of a semispectral measure having the structure of the convolution of a 
positive measure and a semispectral measure are studied, with paying attention to the natural domains of these unbounded operators. 
The results are then applied to conveniently determine the moment operators of the Cartesian margins of the phase space observables.

\

\noindent {\bf Keywords:}  moment operators, convolution,
semispectral measure, phase space observables
\end{abstract}
\maketitle

\section{Introduction}
The increasingly accepted view of a quantum observable as a positive operator  measure 
as opposed to the more traditional approach  using only spectral measures 
has added a great deal to our understanding of the mathematical structure and foundational aspects of
quantum mechanics.
%
In many cases  an observable that is not itself projection valued, 
nevertheless arises as an unsharp or smeared version of a spectral measure. One way to realize such a smearing is 
to convolve the spectral measure with a probability measure. 
In particular, the marginal observables of a phase space observable have such a structure.

Phase space observables have several important applications in
quantum mechanics, ranging from the theory of Husimi distributions
in quantum optics  to state tomography,  phase space
quantizations, and to the theory of approximate joint measurements
of position and momentum, as highlighted, for instance,   by the
monographs
\cite{Davies76,Holevo82,OQP95,Schroeck96,Leonhardt97,Stulpe97,Landsman98,Ali00,Holevo01,Paris04,Landsman07}.
Optical implementations of such  observables are also well
understood, as described e.g. in a recent study \cite{JukkaVII},
and their mathematical structure has been investigated with great
detail \cite{Holevo79,Werner84,CDeV03,JukkaII06}. 
 
The moments of measurement statistics of an observable are related to the moment operators of that
observable in the same way as the outcome probabilities are related to the observable itself. In some cases, the moments may even carry the entire information on the observable \cite{Dvurecenskij2000,JukkaVI}. Hence, it makes sense to study the moment operators of a semispectral measure obtained as a convolution, which is the aim of this paper. In Sect. \ref{preliminaries}, we give the technical lemmas needed for the main results. In particular, we discuss the difficulties in the integrability questions associated with convolutions of nonpositive scalar measures. In Sect.~\ref{moments}, we consider the case of a general convolved semispectral measure by using the operator integral of \cite{LMY}, and in Sect.~\ref{phase_space}, we work out the Cartesian marginal moment operators for a class of phase space observables.

\section{Preliminaries}\label{preliminaries}

To begin with, we recall the notion of the convolution of scalar
measures, and we prove a lemma on their moment integrals.


The \emph{convolution} of two complex Borel measures $\mu,\nu:\h
B(\R)\to \C$ is the measure $\mu*\nu:\h B(\R)\to \C$, defined by
$$
\mu*\nu(X)=(\mu\times \nu)(\{(x,y)\mid x+y\in X\}),  \ \ X\in \h
B(\R),
$$
where $\mu\times \nu$ is the product measure defined on $\h
B(\R^2)$, the Borel $\sigma$-algebra of $\R^2$ (see e.g. \cite[p.
648, Definition 8]{Dunford}).

\begin{lemma}\label{integrallemma} Let $\mu,\nu:\h B(\R)\to \C$ be two complex measures, and let
$k\in \N$.
\begin{itemize}
\item[(a)] A Borel function $f:\R\to \C$ is $\mu*\nu$-integrable
if $(x,y)\mapsto f(x+y)$ is integrable with respect to the product
measure $\mu\times \nu:\h B(\R^2)\to \C$. In that case,
$$
\int f(x)\, d(\mu*\nu)(x) = \int f(x+y)\,d(\mu\times \nu)(x,y).
$$
\item[(b)] The function $(x,y)\mapsto (x+y)^k$ is $\mu\times \nu$-integrable,
if and only if $x\mapsto x^k$ is both $\mu$- and $\nu$-integrable.
In that case,
$$
\int (x+y)^k d(\mu\times \nu)(x,y) = \sum_{n=0}^k \binom{k}{n}
\left(\int x^{k-n}d\mu(x)\right) \left(\int y^nd\nu(y)\right).
$$
\end{itemize}
\end{lemma}
\begin{proof} Let $\phi:\R^2\to \R$ denote addition, i.e.
$\phi(x,y)=x+y$.
Write the product measure $\mu\times \nu$ in terms of the positive
and negative parts of its real and imaginary parts:
$$
\mu\times \nu = \nu_1+i\nu_2=\nu_1^+-\nu_1^-+i(\nu_2^+-\nu_2^-),
$$
where $\nu_i^{\pm}=\frac 12 (|\nu_i|\pm \nu_i)$. Then
\begin{equation}\label{sum}
\mu*\nu(X) =
\nu_1^+(\phi^{-1}(X))-\nu_1^-(\phi^{-1}(X))+i(\nu_2^+(\phi^{-1}(X))-\nu_2^-(\phi^{-1}(X))),
\  \ X\in \h B(\R).
\end{equation}
Assume now that $f\circ\phi$ is $\mu\times \nu$-integrable. Then
$f\circ\phi$ is integrable with respect to each $\nu_i^{\pm}$, and
\begin{equation}\label{integral}
\int f\circ \phi\, d(\mu\times \nu) =\int f\circ\phi
\,d\nu_1^+-\int f\circ \phi \,d\nu_1^- +i(\int f\circ \phi
\,d\nu_2^+-\int f\circ \phi \,d\nu_2^-).
\end{equation}
Since the measures $\nu_i^{\pm}$ are positive, it follows that $f$
is integrable with respect to each induced measure $X\mapsto
\nu_1^+(\phi^{-1}(X))$, and the corresponding integrals are equal
(see e.g. \cite[p. 163]{Halmos}). Now \eqref{sum} and
\eqref{integral} imply that $f$ is $\mu*\nu$-integrable, with
$$
\int f d(\mu*\nu) = \int f\circ \phi \,d(\mu\times \nu).
$$
This proves (a).

To prove (b), suppose first that $(x,y)\mapsto \phi(x,y)^k$ is
$\mu\times \nu$-integrable. Since both measures $\mu$ and $\nu$
are finite, it follows from \cite[p. 193, Theorem 13]{Dunford} that
$x\mapsto \phi(x,y)^k=(x+y)^k$ is $\mu$-integrable for
$\nu$-almost all $y\in \R$. Take any such $y\in \R$. Now $x\mapsto
|x+y|^k$ is also $|\mu|$-integrable, where $|\mu|$ denotes the
total variation measure of $\mu$. There are positive constants $M$
and $K$ satisfying
$$
|x^k|\leq K|x+y|^k +M, \ \  x\in \R.
$$
This implies that $x\mapsto |x|^k$ is $|\mu$|-integrable, and
hence also $\mu$-integrable.   It is
  similarly seen that $x\mapsto |x|^k$ is
$\nu$-integrable.

Suppose now that $x\mapsto x^k$ is both $\mu$- and
$\nu$-integrable. Since $x\mapsto |x|^k$ is now $|\mu|$- and
$|\nu|$-integrable, and these are finite positive measures, it
follows that $x\mapsto |x|^l$ is $|\mu|$- and $|\nu|$-integrable
for all $l\in \N$, $l\leq k$. Hence, $(x,y)\mapsto |x^ly^m|$ is
$|\mu|\times |\nu|$-integrable for all $l,m\in \N$, $l\leq k$,
$m\leq k$. Since $|x+y|^k\leq \sum_{n=0}^k \binom{k}{n}
|x^{k-n}y^n|$, this implies that $(x,y)\mapsto
|\phi(x,y)^k|=|(x+y)^k|$ is $|\mu|\times |\nu|$-integrable. But
$|\mu|\times |\nu|=|\mu\times \nu|$ by \cite[p. 192, Lemma
11]{Dunford}, so $x\mapsto \phi(x,y)^k$ is $\mu\times
\nu$-integrable.

The claimed formula follows now easily, since we have shown above
that the equivalent integrability conditions imply that
$(x,y)\mapsto x^ly^m$ is $\mu\times \nu$-integrable for all
$l,m\in \N$, $l\leq k$, $m\leq k$.

\end{proof}

The converse implication in part (a) of the above lemma does not
hold if the measures $\mu$ and $\nu$ are not assumed to be
positive. This is the conclusion of the brief discussion we now
enter. Denote $\phi(x+y)=x+y$ as before, and
$\Sigma=\{\phi^{-1}(X)\,|\,X\in \hB(\R)\}.$ Then $\Sigma$ is a
$\sigma$-algebra (properly) contained in $\hB(\R^2)$.  Let $\mu$
and $\nu$ be complex Borel measures on
$\R$ and $\lambda_2$ their convolution. 
Lemma 8 in \cite[p. 182]{Dunford} states that the formula
$\lambda_1(\phi^{-1}(X))=\lambda_2(X)$ gives a well-defined
complex measure on $\Sigma$. (To be precise, the lemma requires
the additional assumption that $\phi$ be surjective, but of course
this holds in our situation.) The same lemma says that the total
variations satisfy $|\lambda_1|(\phi^{-1}(X))=|\lambda_2|(X)$ for
all $X\in \hB(\R)$, and moreover for any $\lambda_2$-integrable
Borel function $f:\R\to\C$, the composite function $f\circ \phi$
is $\lambda_1$-integrable, and the natural integral transformation
formula holds.

Obviously $\lambda_1$ is just the restriction of the product
measure $\mu\times\nu$  to $\Sigma$. Any $\Sigma$-measurable
function $g:\R^2\to\C$  which is $\mu\times\nu$-integrable, is
integrable with respect to the restriction of the variation
measure $|\mu\times\nu|$ to $\Sigma$. In the following example we
see that this need not be the case if $g$ is just assumed to be
$\lambda_1$-integrable. This phenomenon is at the root of the fact
that the implication in Lemma \ref{integrallemma} (a) cannot be
reversed.

\begin{example}\label{example1}
\rm   We construct two discrete measures $\mu$ and $\nu$
supported by $\Z$. Let $\sum_{k=0}^\infty a_k$ be any convergent
series with positive terms, and define $b_{2k}=b_{2k+1}=a_k$ for
all $k=0,\,1,\,2,\dots$, and $b_k=0$ if $k\in \Z$, $k\leq -1$. We
set $\mu(\{n\})=b_n$ for all $n\in\Z$. The discrete measure $\nu$
is defined by setting $\nu(\{n\})=(-1)^nb_{-n}$ for all $n\in \Z$.
Then the convolution $\lambda=\mu*\nu$ is supported by $\Z$, and
we have $\lambda(\{n\})=\sum_{j=-\infty}^\infty
b_j(-1)^{n-j}b_{j-n}=\sum_{j=0}^\infty b_j(-1)^{n-j}b_{j-n}$. If
$n$ is even, it follows that $\lambda(\{n\})=0$, since
$b_{2k}=b_{2k+1}$. However, $c_n=\sum_{j=-\infty}^\infty
|b_j(-1)^{n-j}b_{j-n}|>0$. We now define $f:\R\to\C$ by setting
$f(2k)=c_{2k}^{-1}$ for all $k\in Z$ and $f(x)=0$ if
$x\in\R\setminus 2\Z$. Then $\int_\R f(x)d\lambda(x)=0$, but the
function $(x,y)\mapsto f(x+y)$ is not $|\mu\times\nu|$-integrable,
since its integral with respect to $|\mu\times\nu|$ over any set
$\{(x,y)\,|\,x+y=n\}$, $n\in 2\Z$, equals 1.

\end{example}


To close this preliminary section, we recall the notion of an
operator integral in the sense of  \cite{LMY}. Let $\Omega$ be a
set and $\hA$ a $\sigma$-algebra of subsets of $\Omega$. Let
$E:\hA\to\lh$ be a semispectral measure (normalized positive
operator measure) taking values in $\lh$, the set of
bounded operators on a complex 
Hilbert space $\hil$ ($\ne \{0\}$). Thus, for any
$\vp,\psi\in\hil$, the set function $X\mapsto E_{\psi,\vp}(X):=
\langle\psi\,|\, E(X)\vp\rangle$ is a complex measure. For any
measurable function $f:\Omega\to\C$ we let $D(f,E)$ denote the set
of those vectors $\vp\in\hil$ for which $f$ is
$E_{\psi,\vp}$-integrable for all $\psi\in\hil$. The set $D(f,E)$
is a vector subspace of $\hil$ and the formula
$$
 \langle\psi\,|\, L(f,E)\vp\rangle =\int_\Omega f\,dE_{\psi,\vp}, \quad \vp\in D(f,E), \psi\in\hil,
$$
defines a unique linear operator $L(f,E)$, with the domain
$D(f,E)$. The set $\tilde D(f,E)=\{\vp\in\hil\,|\, \int
|f|^2\,dE_{\vp,\vp} <\infty \}$ is a subspace of $D(f,E)$, and we
let $\tilde L(f,E)$ denote the restriction of $L(f,E)$ into
$\tilde D(f,E)$. We recall that if $E$ is a spectral (projection
valued) measure, then $\tilde D(f,E)=D(f,E)$ and the operator
$L(f,E)$ is densely defined. We consider here only the cases where
$(\Omega,\hA)$ is $(\R,\br)$ or  $(\R^2,\hB(\R^2))$.

\section{Convolutions and their moment operators}\label{moments}

For any $X\in \h B(\R)$, let $\chi_X$ denote the characteristic
function of  $X$.  Recall that $\phi$ denotes the map
$(x,y)\mapsto x+y$. The function $\chi_X\circ \phi$ is bounded and
thereby integrable with respect to the product measure. Hence
Lemma~\ref{integrallemma}(a) and Fubini's theorem give that the
function
$$y\mapsto \mu(X-y)=\int\chi_{X-y}(x)\, d\mu(x)=\int\chi_X(x+y)\, d\mu(x)$$
coincides almost everywhere with a Borel function, and
$$
\mu*\nu(X) = \int\left(\int \chi_X(x+y)\, d\mu(x)\right)\, d\nu(y)
= \int_\R\mu(X-y)d\nu(y), \  \ X\in \h B(\R).
$$
Let now $E:\h B(\R)\to L(\hil)$ be a semispectral measure, and let
$\mu:\h B(\R)\to [0,1]$ be a probability measure. Since the
sesquilinear form
$$(\vp,\psi)\mapsto \int_\R \mu(X-y)\, dE_{\psi,\vp}(y)$$
is clearly bounded, one can define $\mu*E:\h B(\R)\to L(\hil)$ via
$\langle \vp|(\mu*E)(X)\psi\rangle := \mu*E_{\psi,\vp}(X)$,
$\vp,\psi\in \hil$. It follows from the monotone convergence
theorem that $\mu*E$ is a semispectral measure.

Denote $\mu[k]:=\int x^k\, d\mu(x)$, in case this integral exists
(i.e. when $\int |x^k|\, d\mu(x)<\infty$).

\begin{proposition} Let $E:\h B(\R)\to L(\hil)$ be a semispectral
measure, and $\mu:\h B(\R)\to [0,1]$ a probability measure. Then
\begin{itemize}
\item[(a)] $\tilde{D}(x^k,\mu*E)$ equals either $\tilde{D}(x^k,E)$ or
$\{0\}$, depending on whether $\mu[2k]$ exists or not. In the
former case,
$$\tilde{L}(x^k,\mu*E) = \sum_{n=0}^k \binom{k}{n} \mu[k-n]\tilde{L}(x^n,E).$$
\item[(b)] If $\mu[k]$ exists, then $D(x^k,E)\subset D(x^k,\mu*E)$, and
$$L(x^k,\mu*E) \supset \sum_{n=0}^k \binom{k}{n} \mu[k-n]L(x^n,E).$$
\end{itemize}
\end{proposition}
\begin{proof} Since $(\mu*E)_{\vp,\vp}= \mu*E_{\vp,\vp}$ by definition,
and these measures are positive, it follows from e.g. \cite[p.
163]{Halmos} that $x^{2k}$ is $(\mu*E)_{\vp,\vp}$-integrable if
and only if $(x,y)\mapsto (x+y)^{2k}$ is $\mu\times
E_{\vp,\vp}$-integrable. By Lemma \ref{integrallemma} (b), this
happens exactly when $\mu[2k]$ exists and $\vp\in
\tilde{D}(x^k,E)$. Hence, $\tilde{D}(x^k,\mu*E)$ equals either
$\tilde{D}(x^k,E)$ or $\{0\}$, depending on whether $\mu[2k]$
exists or not. Suppose now that $\mu[2k]$ exists, and let $\vp\in
\tilde{D}(x^k,\mu*E)=\tilde{D}(x^k,E)$. Since this set is
contained in $D(x^k,E)$, it follows that $x^k$ is
$E_{\psi,\vp}$-integrable for all $\psi\in \hil$. Also, $\mu[k]$
clearly exists. Hence, according to Lemma \ref{integrallemma} (b),
$(x,y)\mapsto (x+y)^k$ is $\mu\times E_{\psi,\vp}$-integrable for
all $\psi\in \hil$, so using both (a) and (b) of that lemma, we
get
\begin{equation}\label{binomialform}
\int x^k \, d(\mu*E)_{\psi,\vp} = \int (x+y)^k\, d(\mu\times
E_{\psi,\vp})(x,y) = \sum_{n=0}^k \binom{k}{n} \mu[k-n]\int x^n \,
dE_{\psi,\vp}, \  \ \psi\in \hil.
\end{equation}
This completes the proof of (a). To prove (b), suppose that
$\mu[k]$ exists, so that $x^k$ is $\mu$-integrable. Now if $\vp\in
D(x^k,E)$, then $x^k$ is also $E_{\psi,\vp}$-integrable for any
$\psi\in \hil$. According to Lemma \ref{integrallemma} (b), this
implies that $(x,y)\mapsto (x+y)^k$ is $\mu\times
E_{\psi,\vp}$-integrable for all $\psi\in \hil$, and using again
also Lemma \ref{integrallemma} (a), we see that $x^k$ (and thus
also $x^n$ with $n\leq k$) is $\mu*E_{\psi,\vp}$-integrable (i.e.
$(\mu*E)_{\psi,\vp}$-integrable) for all $\psi\in \hil$, and the
relation \eqref{binomialform} holds. But this means that we have
proved (b).
\end{proof}


\begin{proposition}\label{domainprop} Let $E:\h B(\R)\to L(\hil)$ be a spectral measure,
let $k\in \N$, and let $\mu:\h B(\R)\to [0,1]$ be a probability
measure such that $\mu[k]$ exists. Denote $A=L(x,E)$. Then
$$L(x^k,\mu*E)= \sum_{n=0}^k \binom{k}{n} \mu[k-n]A^n, \ \ D(x^k,\mu*E)=D(A^k).$$
Moreover, $\tilde{D}(x^k,\mu*E)$ equals either
$D(A^k)=D(x^k,\mu*E)$ or $\{0\}$, depending on whether $\mu[2k]$
exists or not.
\end{proposition}
\begin{proof} Since $E$ is a spectral measure, $A$ is selfadjoint, and
$D(A^k) = D(x^k,E) = \tilde{D}(x^k,E)$, $L(x^k,E)=A^k$ for all
$k\in \N$. According to the preceding proposition (b),
$L(x^k,\mu*E)$ is a symmetric extension of the selfadjoint
operator $\sum_{n=0}^k \binom{k}{n} \mu[k-n]A^n$. Thus these
operators must be equal. The last claim follows immediately from
part (a) of the preceding proposition.
\end{proof}

\begin{remark}\rm
Let $E:B(\R)\to L(\hil)$ be any spectral measure, and choose a
probability measure $\mu$ such that $\mu[k]$ exists but $\mu[2k]$
does not. Then $L(x^k,\mu*E)$ is a densely defined selfadjoint
operator, but $\tilde{D}(x^k,\mu*E)=\{0\}$.
\end{remark}

Consider then the following special case. For any positive
operator $T$ of trace one, and a selfadjoint operator $A$ in
$\hil$, let $p_T^A:\h B(\R)\to [0,1]$ be the probability measure
defined by $p_T^A(X) = \tr[TE^{A}(X)]$, where $E^A$ is the
spectral measure of $A$.

Let $A$ be a selfadjoint operator and $k\in \N$, such that
$p_T^A[k]$ exists. According to e.g. \cite[Lemma 1]{KLY2005} and
\cite[Lemma 1]{KLY2006}, this happens exactly when
$\sqrt{|A|}^k\sqrt{T}$ is a Hilbert-Schmidt operator. Under this
condition, we then have, according to the preceding proposition,
that
$$
L(x^k,p_T^{A}*E^B)=\sum_{n=0}^k \binom{k}{n} p_T^A[k-n]B^n, \ \
D(x^k,p_T^A*E^B)=D(B^k)
$$
for any selfadjoint operator $B$. Moreover,
$\tilde{D}(x^k,p_T^{A}*E^B)\neq \{0\}$ if and only if $p_T^A[2k]$
exists, or, equivalently, $A^k\sqrt{T}$ is a Hilbert-Schmidt
operator. This stronger condition assures also that
$p_T^A[k-n]=\tr[A^{k-n}T]$ in the above formula, the operators
$A^{k-n}T$ being in the trace class.

\begin{remark}\rm As an example, take
$T=|\eta\rangle\langle \eta|$ with $\eta\in D(\sqrt{|A|})$ but
$\eta\notin D(A)$. Then
$L(x,p_{|\eta\rangle\langle\eta|}^{A}*E^B)=B$, since $\sqrt{|A|}
\sqrt{|\eta\rangle \langle \eta|} = \sqrt{|A|}|\eta\rangle \langle
\eta|$ is clearly a Hilbert-Schmidt operator. However, the square
integrability domain is $\{0\}$, since
$|A|\sqrt{|\eta\rangle\langle \eta|}$ is quite far from being
Hilbert-Schmidt (its domain is $\{0\}$). Note also that now
$p_{|\eta\rangle\langle \eta|}^A[1]$ is not equal to
$\tr[A|\eta\rangle\langle \eta|]$, since this trace is not even
defined.
\end{remark}

\section{Phase space observables}\label{phase_space}
Let $\hil= L^2(\R)$, and let $Q$ and $P$ be the selfadjoint
position and momentum operators
 in $\hil$, and  $W(q,p)$,
$(q,p)\in\R^2$, the corresponding Weyl operators. Consider now the
phase space observable $E^T:\h B(\R^2)\to L(\hil)$,
$$
E^T(Z)=\frac{1}{2\pi}\int_Z W(q,p)TW(q,p)^*\, dqdp,
$$
with $T$ a positive operator of trace one.  The Cartesian marginal
measures $E^{T,x},E^{T,y}:\h B(\R)\to L(\hil)$ are defined by
$E^{T,x}(X):=E^T(X\times \R)$, $E^{T,y}(Y):=E^T(\R\times Y)$. It
is well known that they are equal to $p_T^{-Q}*E^Q$ and
$p_T^{-P}*E^P$, respectively, see  e.g.  \cite[Theorem
3.4.2]{Davies76}. According to the above discussion, we can thus
determine the $k$th moment operators of the $x$- and $y$- margins,
under the respective conditions that $p_T^{-Q}[k]$ and
$p_T^{-P}[k]$ exist, or, equivalently, $\sqrt{|Q|}^k\sqrt{T}$ and
$\sqrt{|P|}^k\sqrt{T}$ are Hilbert-Schmidt:

\begin{proposition} Let $k\in \N$.
\begin{itemize}
\item[(a)] If $\sqrt{|Q|}^k\sqrt{T}$ is a Hilbert-Schmidt operator, then
$$
L(x^k,E^{T,x}) = \sum_{n=0}^k \binom{k}{n}
(-1)^{k-n}p_T^{Q}[k-n]Q^n, \ \ D(x^k,E^{T,x})=D(Q^k).
$$
\item[(b)] Part (a) holds also when ''$x$'' and ''$Q$'' are replaced by
''$y$'' and ''$P$''.
\end{itemize}
\end{proposition}
Under the square integrability condition that $Q^k\sqrt{T}$
(respectively $P^k\sqrt{T}$) be Hilbert-Schmidt, we get
$p_T^Q[k-n]= \tr[Q^{k-n}T]$ ($p_T^P[k-n]= \tr[P^{k-n}T]$).

\begin{remark}\rm According to the discussion in the preceding
remark, a simple example where $L(x,E^{T,x})=Q$ but
$\tilde{D}(x,E^{T,x})=\{0\}$, is obtained by taking
$T=|\eta\rangle \langle \eta|$, where $\eta\in \hil$
is a unit vector with $\int |x||\eta(x)|^2\, dx<\infty$, $\int x
|\eta(x)|^2\, dx = 0$, and $\int x^2|\eta(x)|^2\, dx=\infty$.
\end{remark}

An additional problem with the moment operators $L(x^k,E^{T,x})$
and $L(x^k,E^{T,y})$ is their connection to the operators
$L(x^k,E^T)$ and $L(y^k,E^T)$, which we have considered before
(see \cite{KLY2005,KLY2006}). By writing e.g.
$E^{T,x}(X)=E^T(\pi_1^{-1}(X))$ where $\pi_1:\R^2\to \R$ is the
coordinate projection $(x,y)\mapsto x$, we notice that a similar
''change of variables'' argument as that in Lemma
\ref{integrallemma} gives $L(x^k, E^{T,x})\supset L(x^k,E^T)$.
 Now if $\sqrt{|Q|}^k\sqrt{T}$ is
Hilbert-Schmidt, then we know from the above proposition that
$L(x^k, E^{T,x})$ is a selfadjoint operator, a polynomial in $Q$.
However, this does not determine $L(x^k,E^T)$; we can only say
that it has $L(x^k, E^{T,x})$ as a selfadjoint extension.

Consider then the square integrability domains. Since the measures
involved are now positive, the ''change of variables formula''
(see e.g. \cite[p. 163]{Halmos}) can be used to conclude that the
restrictions are equal: $\tilde{L}(x^k, E^{T,x})=
\tilde{L}(x^k,E^T)$. According to Proposition \ref{domainprop},
this operator is nontrivial exactly when $|Q|^k\sqrt{T}$ is
Hilbert-Schmidt, in which case it is selfadjoint. This stronger
condition then forces both the symmetric extensions $L(x^k,
E^{T,x})$ and $L(x^k,E^T)$ to coincide with the restriction, and
we recover Theorem 4 of \cite{KLY2006}.

\

\noindent {\bf Acknowledgment.} One of us (J.K.) was supported by the Emil Aaltonen Foundation and the Finnish Cultural Foundation.


\begin{thebibliography}{99}
\bibitem{Davies76} E. B. Davies; Quantum Theory of Open Systems, Academic Press, London, 1976.
\bibitem{Holevo82} A.S. Holevo, {\em Probabilistic and Statistical Aspects of Quantum Theory}, North-Holland, Amsterdam, 1982.
\bibitem{OQP95} P. Busch, M. Grabowski, P. Lahti; Operational Quantum Physics, 2nd Corrected Printing, Springer-Verlag,
Berlin, 1997.
\bibitem{Schroeck96} F. E. Schroeck, Jr., {\em Quantum Mechanics on Phase Space}, Kluwer Academic Publishers, Dordrecht, 1996.
\bibitem{Leonhardt97} U. Leonhardt, {\em Measuring the Quantum State of Light}, Cambridge University Press, Cambridge, 1997.
\bibitem{Stulpe97} W. Stulpe; {\em Classical Representations of Quantum Mechanics Related to
Statistically Complete Observables}, Wissenschaft und Technik
Verlag, Berlin, 1997.
\bibitem{Landsman98} N. P. Landsman; {\em Mathematical Topics Between Classical and Quantum Mechanics}, Springer-Verlag,
New York, 1998.
\bibitem{Ali00} S. T. Ali, J.-P. Antoine, J.-P. Gazeau; {\em Coherent States, Wavelets and Their Generalizations}, Springer-Verlag,
New York, 2000.
\bibitem{Holevo01} A. S. Holevo, {\em Statistical Structure of Quantum Theory}, Springer, Berlin, 2001.
\bibitem{Paris04} M. G. A. Paris, J. \v Reh\'a\v cek (Eds.), {\em Quantum State Estimation}, Lect. Notes Phys. {\bf 649},
Springer-Verlag, Berlin, 2004.
\bibitem{Landsman07} N. P. Landsman; {\em Between Classical and Quantum},
Handbook of the Philosophy of Science, Philosophy of Physics,
Elsevier, 2007.
\bibitem{JukkaVII}  J. Kiukas, P. Lahti,
A note on the measurement of phase space observables with an
eight-port homodyne detector, {\em J. Mod. Optics} {\bf 55} (2008)
1891-1898.
\bibitem{Holevo79} A. S. Holevo, Covariant measurements and uncertainty relations,
{\em Rep. Math. Phys. }{\bf 16} (1979) 385-400.
\bibitem{Werner84} R. Werner; Quantum harmonic analysis on phase space, {\em J. Math. Phys.} {\bf 25} (1984) 1404-1411.
\bibitem{CDeV03} G. Cassinelli, E. De Vito, A. Toigo; Positive operator valued measures covariant
with respect to an irreducible representation, {\em J. Math. Phys.} {\bf
44} (2003) 4768-4775.
\bibitem{JukkaII06} J. Kiukas, P. Lahti, K. Ylinen, Normal covariant quantization maps, {\em J. Math. Anal. Appl.} {\bf 319} 783-801 (2006).


\bibitem{Dvurecenskij2000} A. Dvure\v{c}enskij, P. Lahti, K. Ylinen, Positive operator measures determined by their moment sequences, {\em Rep. Math. Phys.} {\bf 45} 139-146 (2000).
\bibitem{JukkaVI}
J. Kiukas, P. Lahti, On the moment limit of quantum observables, with an application to the balanced homodyne detection, {\em Journal of Modern Optics}, {\bf 55} 1175-1198 (2008).
\bibitem{LMY} P. Lahti, M. Maczy\'nski, K. Ylinen; The moment operators of
phase space observables and their number margins, Rep. Math. Phys.
{\bf 41} (1998) 319-331.



\bibitem{Dunford} N. Dunford, J. T. Schwartz, {\em Linear Operators, Part I: General Theory}, Interscience Publishers, New York, 1958.
\bibitem{Halmos} P. R. Halmos, Measure Theory, Springer-Verlag, New York, 1974,; originally published in 1950.
\bibitem{KLY2005} J. Kiukas, P. Lahti, K. Ylinen, Moment operators of the Cartesian margins of the phase space observables,
{\em J. Math. Phys.} {\bf 46} 042107 (2005).
\bibitem{KLY2006} J. Kiukas, P. Lahti, K. Ylinen, Phase space quantization and
the operator moment problem, {\em J. Math. Phys.} {\bf 47} 072104
(2006).
\end{thebibliography}
\end{document}